\documentclass[10pt]{amsart}
\usepackage[cp1251]{inputenc}
\usepackage[english]{babel}
\usepackage{amsmath}
\usepackage{amssymb}
\usepackage{amsfonts}
\usepackage{srcltx} 
\usepackage[dvips]{graphicx}

\newtheorem{lemma}{Lemma}
\newtheorem{theorem}{Theorem}

\newtheorem{corollary}{Corollary}

\begin{document}

\title{On  transitive uniform partitions of
$F^n$ \\ into binary Hamming codes}
\author{{F. I. Solov'eva}}%

\address{Faina I. Solov'eva
\newline\hphantom{iii} Sobolev Institute of Mathematics,
\newline\hphantom{iii} pr. ac. Koptyuga 4,
\newline\hphantom{iii} 630090, Novosibirsk, Russia}%
\email{sol@math.nsc.ru}%

\thanks{\copyright \ 2018   F. I. Solov'eva}
\thanks{\rm The work has been
supported by RFBR Grant 19-01-00682.}



\maketitle

\begin{quote}
{\small \noindent{\sc Abstract. } We investigate transitive
uniform partitions of the vector space $F^n$ of dimension $n$ over
the Galois field $GF(2)$ into cosets of Hamming codes. A partition
$P^n= \{H_0,H_1+e_1,\ldots,H_n+e_n\}$ of $F^n$ into cosets of
Hamming codes  $H_0,H_1,\ldots,H_n$ of length $n$ is said to be
uniform if the intersection of any two codes $H_i$ and $H_j$,
$i,j\in \{0,1,\ldots,n \}$ is constant, here $e_i$ is a binary
vector in $F^n$ of weight $1$ with one in the $i$th coordinate
position.
 For any $n=2^m-1$, $m>4$ we found a  class of nonequivalent  $2$-transitive uniform partitions of  $F^n$ into cosets of Hamming
codes.

\medskip

}
\end{quote}

\section{Introduction}

By  $F^n$ we denote the vector space of dimension $n$ over the
Galois field $GF(2)$ with respect to the Hamming metric. In this
short correspondence using a recursive construction we prove the
existence of nonequivalent $2$-transitive uniform partitions of
$F^n$ into cosets of Hamming codes for any length $n=2^m-1$,
$m>4$.

 The {\em Hamming distance} $d(x,y)$ between
any two vectors $x,y\in F^n$ is defined as the number of
coordinates in which $x$ and $y$ differ. {\it The Hamming weight}
$w(x)$ of a vector $x$ is $d(x,{\bf 0}^n)$, where ${\bf 0}^n$ is
the all-zero vector of length  $n$. A {\it code} of length $n$ is
a subset of $F^n$, its elements  are called {\it codewords}. The
{\it code distance} of a code is the minimum value of the Hamming
distance between two different codewords from the code. A code $C$
is called {\it perfect binary single-error-correcting}
 (briefly  {\it perfect}) if for any vector $x$ from the set $F^n$  there
 exists exactly one vector $y
\in C$ at the Hamming distance not more than 1 from the vector
$x$. A perfect linear code is called the {\it Hamming code}.

{\it The automorphism group of a partition} $P^n=
\{C_0,C_1,\ldots,C_n\}$ of length $n$ of  $F^n$ into perfect codes
$C_0,C_1,\ldots,C_n$ of length $n$, $\bigcup\limits_{i=0}^{n}C_i =
F^n$ is defined as the group of all isometries of $F^n$ preserving
the partition $P^n$. A partition $P^n$ is called {\it transitive},
if for any two codes $ C_i$ and $  C_j$, $i, j$ from $I=
\{0,1,\ldots,n \},$ there is an automorphism $\sigma$ from $
\mbox{Aut}(P^n)$ such that $\sigma(C_i) = C_j.$  A partition $P^n$
of $F^n$ is defined to be {\it $2$-transitive}, if for any two
subsets $\{i_1,i_2\}$ and $\{j_1,j_2\}$ of $I$ there exists an
automorphism $\sigma$ from $Aut(P^n)$ such that $\sigma(C_{i_t}) =
C_{j_t}, t=1,2.$  By definition any $2$-transitive partition is
transitive.
 Let $e_i$
be a binary vector in $F^n$ of weight $1$ with one in the $i$th
coordinate  position. A partition $P^n=
\{H_0,H_1+e_1,\ldots,H_n+e_n\}$ of $F^n$ into cosets of Hamming
codes  $H_0,H_1,\ldots,H_n$ of length $n$  is said to be {\it
uniform} if the intersection of any two  codes $H_i$ and $H_j$,
$i,j\in I$ is constant, the number $\log\eta_n$ is called the {\it
uniformity number}. Here and below $\log$ stands for the binary
logarithm. Note that transitive partitions are not necessarily
uniform.

In~\cite{HS09} several classes of partitions of the space $F^n$
into mutually nonparallel cosets of Hamming codes were presented
and the lower bound on the number of nonequivalent partitions was
given. The partitions in these constructions were not generally
transitive. In~\cite{SG2011}  constructions of transitive,
vertex-transitive and $2$-transitive partitions of $F^n$ into
perfect codes and lower bounds on the number of nonequivalent such
partitions  were given. Note that the  partitions presented
in~\cite{HS09,SG2011} were not uniform.
 Uniform  partitions of $F^n$ into cosets
of Hamming codes and into extended Hamming codes with the smallest
possible  $\eta_n$ were constructed for length $n=7$ by Phelps
in~\cite{Phelps} and for any $n=2^m$ for odd $m>3$, using the Gold
function by Krotov in~\cite{Krotov}. In \cite{S2007} partitions
into pairwise nonparallel Hamming codes were constructed and among
them there were uniform partitions.  An overview of  results till
1998 on utilizing partitions to construct $q$-ary perfect codes
can be found in~\cite{CHLL}. See~\cite{Sol} for a survey
concerning some new results on partitions and all other necessary
definitions and notions.

The investigation of the partitions of $F^n$ into perfect codes is
 important due to the connection of the
classification problem of all partitions with the analogous
problem for perfect binary codes. It is known that  the limit for
the relation of double logarithms of the numbers of different
perfect binary codes and different partitions equals $1$, although
the number of nonequivalent partitions significantly exceeds the
number of nonequivalent codes.  Since an extended   Hamming code
of length $n$ is  the  Reed -- Muller code of the same length and
order $n-2$, see \cite{MWSl}, it is natural to investigate the
problem of constructing partitions of $F^n$ into Reed -- Muller
codes of any admissible order.  The intersection of two Hamming
codes often gives a good cyclic code, see, for example,
\cite{MWSl}. Moreover partitions are connected with the  perfect
colorings called also regular codes, partition designs or
equitable partitions~\cite{FdF}. In some cases partitions of $F^n$
into codes induce colorings associated with fibre optic
nets~\cite{Ost}.

\section{Construction}

In order to give a recursive construction of the class of
 uniform partitions into Hamming codes we exploit
the construction B from~\cite{S2007} based on the  classical
Mollard construction for perfect codes~\cite{Mollard} and the
results given in~\cite{Krotov,Phelps,SG2011}.   For the sake of
completeness we recall the definition of the Mollard construction
and the construction B.

Let $C^l$ and $C^t$
 be two   binary codes of lengths $l$
and $t$ respectively with the code distance not less than $3$
containing the all-zero vectors. Let
$$x=(x_{11},x_{12}, \ldots,x_{1t},
x_{21},\ldots,x_{2t},\ldots,x_{l1},\ldots,x_{lt}) \in\,F^{lt}.$$
We  use a matrix notation of vector $x$: the $i$th row of the
matrix  is equal to $x_{i1} \,\, x_{i2} \,\, \ldots \,\, x_{it}$,
where $i=1,\ldots,l$. Functions $p_1(x)$ and $p_2(x)$ are defined
as
$$
p_1 (x) = \Bigg( \sum_{j=1}^t x_{1j},\ldots, \sum_{j=1}^t x_{lj}
\Bigg) \in F^l,
$$
$$
p_{2}(x)=\Bigg(\sum_{i=1}^{l}x_{i1},\ldots,\sum_{i=1}^{l}x_{it}\Bigg)
\in  F^t.
$$
 Let $f$  be an
arbitrary function from $C^l$ to $F^t.$ The set
$$
C^n = \{ (x,y + p_1 (x),z + p_2 (x) + \,f(y))  \mid x\in F^{lt},\;
y\in C^l,\; z \in C^t\}
$$
is a binary {\it Mollard code} of length $n=lt+l+t$ with the  code
distance~3, see~\cite{Mollard}.

Let $C^l$ and $C^t$ be the Hamming codes of lengths $l$ and $t$
denoted as $H^l$ and $H^t$  respectively and  the function $f$ be
the constant function from $H^l$ to ${\bf 0}^t$. Then  we obtain
the Hamming code of length $n=lt+l+t$ by the Mollard construction
$$
 \{ (x,y + p_1 (x),z + p_2 (x)) \mid x\in F^{lt},\;
y\in H^l,\; z \in H^t\}.
$$

Below we use  a partial case of  the {\it construction B }
from~\cite{S2007}. It should be noted that in~\cite{S2007} the
construction~B was described for any binary
single-errors-correcting codes, not necessarily perfect.
 Let $P^l = \big\{
H_0^l,H_1^l+e_1, \ldots,H_l^l+e_l\big\}$ and
$P^t=\big\{H_0^t,H_1^t+e_1, \ldots,H_t^t+e_t\}$
 be arbitrary partitions of
$F^l$ and $F^t$ into cosets of Hamming codes $H_0^l,H_1^l,
\ldots,H_l^l$  and  $H_0^t,H_1^t, \ldots,H_t^t$ respectively. Then
the set of the codes \begin{equation} \label{constrBbase} \big\{
(x,y + p_1 (x),z + p_2 (x))  \mid x\in\,F^{lt},\;
y\in\,H_i^l+e_i,\;  z \in\,H_j^{t}+e_j \big\}, \end{equation} $ 0
\leq i \leq l$, \, $0 \leq j \leq t$ defines a {\it partition
$P^n$ of the space  $F^n$ into cosets of Hamming codes}
 of length $n=lt+l+t$, see~\cite{S2007}. The
 Hamming code  corresponding to (\ref{constrBbase}) we denote by $M^n(H_i^l,H_j^t)$
 to emphasize that this
 Hamming code of length $n$ is  obtained from the  Hamming codes $H^l_i$ and $H^t_j$ of
lengths $l$ and $t$ respectively.

\begin{lemma} (See \cite{SG2011}.) \label{transitivity} The construction B
applied for  $2$-transitive  partitions into perfect codes of
length $l$ and $t$ gives a  $2$-transitive partition into perfect
codes of length $l+t+lt$.
\end{lemma}

The
 partition of  $F^n$ into  cosets of any Hamming code of length
$n$  is  called  {\it trivial}, so in this case the uniformity
number $\log \eta_n$ is equal to the dimension of the code, i.e.
$n-\log (n+1)$.

\begin{lemma}  \label{uniformity} The construction B
applied for  uniform  partitions into Hamming codes of length $l$
and $t$ gives a  uniform partition into Hamming codes of length
$l+t+lt$.
\end{lemma}

\begin{proof}
Let $\log\eta_{l}$ and $\log\eta_{t}$ be the uniformity numbers of
two initial uniform partitions  into Hamming codes of lengths $l$
and $t$ respectively. Then it is easy to see that applying  the
construction B to both of these uniform partitions we obtain
 the uniform partition of length $n=l+t+lt$ with the
uniformity number $\log\eta_{n}= \log\eta_{l} + \log\eta_{t} +
lt.$

\end{proof}

In Lemma 3, see \cite{EV98}, one can find the simple fact that two
perfect codes of length $n$  intersected by  $s$ codewords exist
if and only if there exist two extended perfect codes of length
$n+1$ intersected by  $s$ codewords. Therefore the partitions of
$F^n$ into Hamming codes obtained from the uniform partitions
given in ~\cite{Krotov} by puncturing any coordinate position are
uniform. Recall that
 in~\cite{Krotov} the  uniform partitions of the set of all odd weight vectors of
$F^{n+1}$ into extended Hamming codes are given.  We will further
call such uniform partitions into cosets of  Hamming codes {\it
punctured}.

\begin{lemma} \label{three cases}
There are uniform partitions  of $F^n$ into cosets of Hamming
codes  of length \, $n\in \{31, \, 127, 1023\}$ \, with the
uniformity numbers \, $\log\eta_{31}=24$, \, $\log\eta_{127}=116$
and $\log\eta_{1023}=1007$ respectively.
\end{lemma}

\begin{proof}
 The case $\log\eta_{31}=24$ is covered applying  Lemma \ref{uniformity}
 to
the trivial partition for $l=3$ with $\log\eta_3=1$ and the
uniform partition from~\cite{Phelps}  for $t=7$ having
$\log\eta_7=2$. So  we have the uniform partition of length $31$
with the uniformity number $\log \eta_{31} = 1+ 2 +21= 24.$

The case $\log\eta_{127}=116$  is achieved by  Lemma
\ref{uniformity} utilizing the  trivial partition for $l=3$ with
$\log\eta_3=1$ and the punctured uniform partition for $t=31$
given in~\cite{Krotov} with the uniformity number
$\log\eta_{31}=22$: $\log \eta_{127} = 1+ 22 +93= 116.$

For the last case we need the uniform partition of length $255$
with $\log\eta_{255}=241$ that can be obtained by Lemma
\ref{uniformity} applying to the partition of length $l=7$
from~\cite{Phelps} with the uniformity number $\log\eta_7=2$ and
the punctured uniform partition for $t=31$ taken
from~\cite{Krotov} with $\log\eta_{31}=22$: $\log \eta_{255} = 2+
22 + 217= 241.$ Then in order to obtain $\log\eta_{1023}=1007$ for
$n=1023$ we apply again Lemma \ref{uniformity} to  the trivial
partition with $l=3$, $\log\eta_3=1$  and the obtained uniform
partition for  $t=255$ having the uniformity number
$\log\eta_{t}=241$: $\log \eta_{1023} = 1+ 241 + 765= 1007.$

\end{proof}

The $2$-transitivity of the uniform partition of length $7$ with
the uniformity number $\eta_{7}=2$ from \cite{Phelps} was proved
in \cite{SG2011}, see Lemma 1.
 The $2$-transitivity of the  uniform
partition of $F^{31}$ considered in Lemma \ref{three cases}
follows from this $2$-transitive uniform partition of length $7$
and $\eta_{7}=2$ and Lemma \ref{transitivity}, so the following
holds

\begin{corollary} \label{2-transitivity_cor}
There is a $2$-transitive uniform partition  of $F^{31}$ into
cosets of Hamming codes  of length \, $31$ \, with the uniformity
number \, $\log\eta_{31}=24$.
\end{corollary}

\begin{theorem} \label{Uniform partitions}
For any $n=2^m-1$, $m>2$ and $e=1,2,\ldots,[(m+1)/2]$,  with the
exception of $m=4$, $e=1$, there exists a uniform partition  of
$F^n$ into cosets of Hamming codes
 of length $n$ with $\eta_n$
 satisfying
 \begin{equation}\label{eta} \log \eta_n=n-2m+2e-\delta(m),\end{equation}
where  $\delta(m)= \left\{ \begin{array}{l} 1\,\,\mbox{for $m \equiv  1 \pmod 2;$}\\
 0 \,\, \mbox{for $m \equiv  0 \pmod 2 .$}\end{array}\right.$
 \end{theorem}
\begin{proof}
 The proof will be done by induction on $m, m\geq 3$ exploiting  the
construction B, see Lemma \ref{uniformity}, and Lemma \ref{three
cases}. In the construction B  we fix the first uniform partition
$P^l$ with $l=3$ and $\log \eta_7=2$ from~\cite{Phelps} and vary
the second partition $P^t$, $t=2^{m-3}-1$ to be any uniform
partition of length $t$ including the trivial one. So the
corresponding Hamming codes are $M^n(H_i^7,H_j^{2^{m-3}-1}),$ \,
$i=0,1,2,\ldots,7, \,\, j=0,1,2,\ldots , 2^{m-3}-1.$  The approach
is valid with the exception of three special cases:
$\eta_{31}=24$, \, $\eta_{127}=116,$ \, $\eta_{1023}=1007$, where
$l=3$ and $\log \eta_7=1$ that were covered by Lemma \ref{three
cases}.

By the specification of the construction B  we need for induction
base the following three initial cases: $m=3,4$ and $5$.   For
$m=3$, i.e. $n=7$ there exist only two uniform partitions
described by Phelps in~\cite{Phelps}. The first one has the
uniformity number $\log \eta_7=2$, the second one is the trivial
partition into the cosets of any Hamming code of length $7$ with
the uniformity number $\log \eta_7=4.$ For $m=4$ we know only the
trivial partition. For $m=5$ there are known three nonequivalent
uniform partitions of length $31$: the punctured partition from
\cite{Krotov} with the uniformity number $22$, the partition with
the uniformity number $24$ from Lemma \ref{three cases} and the
trivial one.

Let the statement of the theorem hold for any number not greater
than $m-1$, so for partitions of length $2^{m-1}-1$. We prove that
it is valid for $m$, i.e. for  $n=2^m-1$. In the induction steps
of the proof of the theorem we will also use  trivial
 partitions into the cosets of any Hamming code of length
 $t$ by the reason that if one partition in the construction B
  is nontrivial uniform partition and another one is  trivial
 then the resulting partition will be   uniform  nontrivial.

 \smallskip

Let  $m$ be even. So  $m-3$ is  odd and  by (\ref{eta}) and
induction hypothesis  we have $\delta(m-3)=1$ and the following
admissible values for uniformity number  $\log \eta_{t}$ for
partitions of length $t=2^{m-3}-1$:
$$\log \eta_{t}=t-2(m-3)+2e^{\prime}-1, \,\, e^{\prime}=
1,2,\ldots, [(m-2)/2].$$ Applying the construction B and taking
into account that $n=7+t +7t$ we calculate the uniformity number
for the partition (\ref{constrBbase}):
\begin{multline*}
\log \eta_n = \log (|M^{n}(H_i^7,H_j^{t})\cap
M^{n}(H_r^7,H_s^{t})|) \\  =2+ (t-2(m-3)+2e^{\prime}-1) +7t
\\ =2 + 8(2^{m-3}-1) - 2(m-3) +2e^{\prime} - 1  =n-2m+2e,\end{multline*}
\noindent where $e=e^{\prime}=1,2,\ldots,(m-2)/2,$ \,\,
$i,r=0,1,2,\ldots,7, \,\, j,s=0,1,2,\ldots , 2^{m-3}-1.$ We obtain
$|M^{n}(H_i^7,H_j^{t})\cap M^{n}(H_r^7,H_s^{t})|$ to be constant
regardless of the choice of \, $i,j,r,s$,  so we construct all the
required uniform partitions of length $n$, i.e. (\ref{eta}), since
there is the trivial partition of length $n$ with the case
$\eta_n=n-m$, i.e. $e=m/2=[(m+1)/2].$

\medskip

Let $m$ be odd. This case  is analogous to the previous one taking
into account that $m-3$ is even and so $\delta(m-3)=0$ in
(\ref{eta}). Therefore we have
$$\log \eta_{t}=t-2(m-3)+2e^{\prime}, \,\mbox{where} \,\, e^{\prime}=
1,2,\ldots, [(m-2)/2].$$  Then for any \,\, $i,r=0,1,2,\ldots,7,
\,\, j,s=0,1,2,\ldots , 2^{m-3}-1$ we have
\begin{multline*}
\log \eta_n = \log (|M^{n}(H_i^7,H_j^{t})\cap
M^{n}(H_r^7,H_s^{t})|) \\ =2+ t-2(m-3)+2e^{\prime} +7t
\,\,\,\,\,\,\,\,\,\,\,\,\,
\\   =2 + 8(2^{m-3}-1) - 2(m-3) +2e^{\prime}=n-2m+2e-1,\end{multline*}
\noindent where $e=e^{\prime}+1=2,3,4,\ldots,(m-1)/2.$
  We obtain (\ref{eta}) adding the first minimal uniformity number given by
the Krotov punctured construction \cite{Krotov}, i.e. the case
when $\log \eta_n =n-2m+1$ with $e=1$, and the trivial uniform
partition into cosets of a Hamming code of length $n$ with
$\eta_n=n-m$, i.e. $e=(m+1)/2$.

\end{proof}

Extended codes or partitions are often objects with larger
automorphism groups, see, for example, \cite{Sol}. It is easy to
see that extending by parity check any uniform partition of length
$n$ gives the uniform partition of the set of all even weight
vectors in $F^{n+1}$. It is clear that any two partitions of $F^n$
into cosets of Hamming codes of length $n$ from Theorem
\ref{Uniform partitions}
 are nonequivalent since
by the construction  they have different uniformity numbers. So
the following holds

\begin{corollary}\label{coro numbers}
For any $n=2^m-1$, $m>2$,  with the exception $m=4$, there exist
at least $[(m+1)/2]$  nonequivalent uniform  partitions of $F^n$
 into cosets
of Hamming codes of length $n$ and of the set of all even weight
vectors in $F^{n+1}$ into cosets of extended Hamming codes of
length $n+1$.
\end{corollary}

Note that some of the partitions constructed in Theorem
\ref{Uniform partitions} are $2$-transitive. The $2$-transitivity
may not hold for the construction B applied to punctured
partitions from \cite{Krotov} as we do not know whether the latter
are $2$-transitive or not.

\begin{theorem} \label{2-transitive Uniform partitions}
For any $n=2^m-1$, $m\geq 6$  there exist at least $[m/3]$
nonequivalent $2$-transitive uniform partitions  of $F^n$ into
cosets of Hamming codes
 of length $n$. For $m=3$ there exist two and for $m=5$ there exist at least two nonequivalent such
 partitions.
  \end{theorem}

 The
proof of this  theorem is the same as that for  Theorem
\ref{Uniform partitions}.  The $2$-transitivity of all obtained
uniform partitions  follows by Lemma \ref{transitivity}, Corollary
\ref{2-transitivity_cor} and the $2$-transitivity of the initial
partition of length $7$ with the uniformity number  $\eta_7=2$
given by \cite{Phelps}. The $2$-transitivity of the latter one was
proved in \cite{SG2011}. Then for $n\geq 6$ there exist at least
$[m/3]$ nonequivalent $2$-transitive partitions among $[(m+1)/2]$
nonequivalent uniform partitions from Theorem \ref{Uniform
partitions}  having different and the largest uniformity numbers.
For $m=3$ there exist two and for $m=5$ there exist at least two
nonequivalent such  partitions.

Extending by parity check a $2$-transitive uniform partition of
length $n$  leads  to the $2$-transitive uniform partition of the
set of all even weight vectors in $F^{n+1}$. In the next corollary
we take into account that partitions in \cite{Krotov} are
$2$-transitive.

\begin{corollary}\label{extending case}
For any $n+1=2^m$, $m\geq 6$ there exist at least $[m/3] +  1$
nonequivalent  $2$-transitive uniform partitions of the set of all
even weight vectors in $F^{n+1}$ into cosets of extended Hamming
codes of length $n+1$. For $m=3$ there exist two and for $m=5$
there exist at least three nonequivalent such
 partitions.
\end{corollary}

{\bf Remarks}. It should be noted that Theorem \ref{Uniform
partitions}  covers a half
 of possible values of the numbers $\eta_n$. Another part is
still open, so the problem of finding of all
 numbers $\eta_n$ is  still open as far as the problem
of the description of all nonequivalent uniform partitions of
$F^n$ into cosets of Hamming codes of length $n$. This part of
 values of the numbers $\eta_n$ could be covered using the
 technique above if we found, for example,  the uniform partition in  $F^{15}$ with the uniformity
number 9 (if such partition exists). It should also be noted that
in the proof of the theorems it is possible to choose  other
variations of the initial partitions into the construction B.
Speaking more preciously one can consider $l>3$, that could give
different (or perhaps nonequivalent) partitions of length $n$ with
the same uniformity number.

\medskip
The author is very grateful  to the anonymous referee for the
useful remarks.


\begin{thebibliography}{8}

\bibitem{CHLL}
Cohen G., Honkala I., Lobstein A., Litsyn S., {\em Covering
codes},
 Elsevier, 1998, pp. 542.

\bibitem{EV98}
  Etzion T., Vardy A.,  On perfect codes and tilings: problems
and solutions,  {\em SIAM J. Discrete Math.}  V. 11, N.~2, 1998,
P.~205\nobreakdash--223.

\bibitem{FdF}
Fon-Der-Flaass D.~G., Perfect 2-colorings of a hypercube, {\em
Siberian Math. J.} V. 48, N.~4,  2007, P.~740\nobreakdash--745.

\bibitem{HS09}
 Heden O.,  Solov'eva F. I., Partitions of $F^n$ into non-parallel
Hamming codes,  {\em Adv. Math. Commun.}, V.  3, N. 4, 2009, P.
385--397.

\bibitem{Krotov} Krotov D.~S., A partition of the hypercube into
maximally nonparallel Hamming codes, {\em Journal of Combinatorial
Designs}, V. 22, N. 4, 2014, P. 179--187.

\bibitem{MWSl} MacWilliams F. J., Sloane N. J. A.,
{\em   The theory of error-correcting codes,
 North-Holland Publishing Company}, 1977,  pp. 762.

 \bibitem{Mollard} Mollard M., A generalized parity function and its use in the
construction of perfect codes, {\em SIAM J. Alg. Discrete Math.},
V. 7, N. 1, 1986, P. 113--115.

\bibitem{Ost} 
\"Osterg{\aa}rd P.~R.~J., On a hypercube coloring problem, {\em
 J.~Combin. Theory}, Ser. A,  V. 108, 2004, P. 199--204.

\bibitem{Phelps} Phelps K.~T.,  An  enumeration of 1-perfect binary codes, {\em  Australas. J. Comb.} V. 21, 2000, P. 287--298.


\bibitem{Sol} Solo{v'}eva F.~I., Survey on perfect codes, {\em  Mathematical
Problems of Cybernetics}, V. 18, 2013, P. 5\nobreakdash--34 (in
Russian).

\bibitem{S2007}
Solo{v'}eva F.~I., On transitive partitions of an $n$-cube into
codes, {\em  Probl. of Inform. Transm.}, V. 45, N. 1, 2009, P.
23--31.



\bibitem{SG2011}
Solo{v'}eva F.~I., Gu{s'}kov G.~K., On construction of
vertex-transitive partitions of n-cube into perfect codes, {\em
Journal of Applied and Industrial Math.},  V. 5, N. 2, 2011, P.
84--100.



\end{thebibliography}
\end{document}